\documentclass[11pt,leqno]{article}
\usepackage[margin=1in]{geometry} 
\usepackage{amssymb,amsfonts,amsmath,bbm,mathrsfs,stmaryrd,mathtools}
\usepackage{xcolor}
\usepackage{url}

\usepackage{extarrows}

\usepackage[shortlabels]{enumitem}
\usepackage{tensor}

\usepackage{xr}
\usepackage[T1]{fontenc}

\usepackage[colorlinks,
linkcolor=black!75!red,
citecolor=blue,
pdftitle={},
pdfproducer={pdfLaTeX},
pdfpagemode=None,
bookmarksopen=true,
bookmarksnumbered=true,
backref=page]{hyperref}

\usepackage{tikz}
\usetikzlibrary{arrows,calc,decorations.pathreplacing,decorations.markings,intersections,shapes.geometric,through,fit,shapes.symbols,positioning,decorations.pathmorphing}

\usepackage{braket}

\usepackage[amsmath,thmmarks,hyperref]{ntheorem}
\usepackage{cleveref}

\newcommand\nthalias[1]{\AddToHook{env/#1/begin}{\crefalias{lemma}{#1}}}

\nthalias{definition}
\nthalias{example}
\nthalias{examples}
\nthalias{remark}
\nthalias{remarks}
\nthalias{convention}
\nthalias{notation}
\nthalias{construction}
\nthalias{theoremN}
\nthalias{propositionN}
\nthalias{corollaryN}
\nthalias{lemma}
\nthalias{proposition}
\nthalias{corollary}
\nthalias{theorem}
\nthalias{conjecture}
\nthalias{question}
\nthalias{assumption}

\creflabelformat{enumi}{#2#1#3}

\crefname{section}{Section}{Sections}
\crefformat{section}{#2Section~#1#3} 
\Crefformat{section}{#2Section~#1#3} 

\crefname{subsection}{\S}{\S\S}
\AtBeginDocument{%
  \crefformat{subsection}{#2\S#1#3}%
  \Crefformat{subsection}{#2\S#1#3}% 
}

\crefname{subsubsection}{\S}{\S\S}
\AtBeginDocument{%
  \crefformat{subsubsection}{#2\S#1#3}%
  \Crefformat{subsubsection}{#2\S#1#3}% 
}

\theoremstyle{plain}

\newtheorem{lemma}{Lemma}[section]

\newtheorem{corollary}[lemma]{Corollary}
\newtheorem{theorem}[lemma]{Theorem}

\theoremstyle{plain}
\theoremnumbering{Alph}

\theoremstyle{plain}
\theorembodyfont{\upshape}
\theoremsymbol{\ensuremath{\blacklozenge}}

\newtheorem{example}[lemma]{Example}

\newtheorem{remark}[lemma]{Remark}

\crefname{definition}{definition}{definitions}
\crefformat{definition}{#2definition~#1#3} 
\Crefformat{definition}{#2Definition~#1#3} 

\crefname{ex}{example}{examples}
\crefformat{example}{#2example~#1#3} 
\Crefformat{example}{#2Example~#1#3} 

\crefname{exs}{example}{examples}
\crefformat{examples}{#2example~#1#3} 
\Crefformat{examples}{#2Example~#1#3} 

\crefname{remark}{remark}{remarks}
\crefformat{remark}{#2remark~#1#3} 
\Crefformat{remark}{#2Remark~#1#3} 

\crefname{remarks}{remark}{remarks}
\crefformat{remarks}{#2remark~#1#3} 
\Crefformat{remarks}{#2Remark~#1#3} 

\crefname{convention}{convention}{conventions}
\crefformat{convention}{#2convention~#1#3} 
\Crefformat{convention}{#2Convention~#1#3} 

\crefname{notation}{notation}{notations}
\crefformat{notation}{#2notation~#1#3} 
\Crefformat{notation}{#2Notation~#1#3} 

\crefname{table}{table}{tables}
\crefformat{table}{#2table~#1#3} 
\Crefformat{table}{#2Table~#1#3}

\crefname{lemma}{lemma}{lemmas}
\crefformat{lemma}{#2lemma~#1#3} 
\Crefformat{lemma}{#2Lemma~#1#3} 

\crefname{proposition}{proposition}{propositions}
\crefformat{proposition}{#2proposition~#1#3} 
\Crefformat{proposition}{#2Proposition~#1#3} 

\crefname{propositionN}{proposition}{propositions}
\crefformat{propositionN}{#2proposition~#1#3} 
\Crefformat{propositionN}{#2Proposition~#1#3} 

\crefname{corollary}{corollary}{corollaries}
\crefformat{corollary}{#2corollary~#1#3} 
\Crefformat{corollary}{#2Corollary~#1#3} 

\crefname{corollaryN}{corollary}{corollaries}
\crefformat{corollaryN}{#2corollary~#1#3} 
\Crefformat{corollaryN}{#2Corollary~#1#3} 

\crefname{theorem}{theorem}{theorems}
\crefformat{theorem}{#2theorem~#1#3} 
\Crefformat{theorem}{#2Theorem~#1#3} 

\crefname{theoremN}{theorem}{theorems}
\crefformat{theoremN}{#2theorem~#1#3} 
\Crefformat{theoremN}{#2Theorem~#1#3} 

\crefname{enumi}{}{}
\crefformat{enumi}{#2#1#3}
\Crefformat{enumi}{#2#1#3}

\crefname{assumption}{assumption}{Assumptions}
\crefformat{assumption}{#2assumption~#1#3} 
\Crefformat{assumption}{#2Assumption~#1#3} 

\crefname{construction}{construction}{Constructions}
\crefformat{construction}{#2construction~#1#3} 
\Crefformat{construction}{#2Construction~#1#3} 

\crefname{question}{question}{Questions}
\crefformat{question}{#2question~#1#3} 
\Crefformat{question}{#2Question~#1#3} 

\crefname{equation}{}{}
\crefformat{equation}{(#2#1#3)} 
\Crefformat{equation}{(#2#1#3)}

\numberwithin{equation}{section}

\theoremstyle{nonumberplain}
\theoremsymbol{\ensuremath{\blacksquare}}

\newtheorem{proof}{Proof}
\newcommand\pf[1]{\newtheorem{#1}{Proof of \Cref{#1}}}

\newcommand\bR{{\mathbb R}}

\newcommand\bZ{{\mathbb Z}}

\newcommand\cA{{\mathcal A}}
\newcommand\cB{{\mathcal B}}
\newcommand\cC{{\mathcal C}}

\newcommand\wt{\widetilde}

\DeclareMathOperator{\id}{id}

\title{Non-representable quantum measures}
\author{Alexandru Chirvasitu}

\begin{document}

\date{}

\newcommand{\Addresses}{{% additional braces for segregating \footnotesize
  \bigskip
  \footnotesize

  \textsc{Department of Mathematics, University at Buffalo}
  \par\nopagebreak
  \textsc{Buffalo, NY 14260-2900, USA}  
  \par\nopagebreak
  \textit{E-mail address}: \texttt{achirvas@buffalo.edu}

  % % \medskip
  % % 
  % % \textsc{Department of Mathematics, INSTITUTION}
  % % \par\nopagebreak
  % % \textsc{ADDRESS}
  % % \par\nopagebreak
  % % \textit{E-mail address}: \texttt{??}
  % % 

}}

\maketitle

\begin{abstract}
  Grade-$d$ measures on a $\sigma$-algebra $\mathcal{A}\subseteq 2^X$ over a set $X$ are generalizations of measures satisfying one of a hierarchy of weak additivity-type conditions initially introduced as interference operators in quantum mechanics. Every signed polymeasure $\lambda$ on $(X,\mathcal{A})^d$ produces a grade-$d$ measure as its diagonal $\widetilde{\lambda}(A):=\lambda(A,\cdots,A)$, and we prove that as soon as $d\ge 2$ measures (as opposed to polymeasures) do not suffice: the separate $\sigma$-additivity of a $\lambda$ producing $\mu=\widetilde{\lambda}$ cannot, generally, be amplified to global $\sigma$-additivity. This amends a result in the literature, asserting the contrary in case $d=2$.
\end{abstract}


\vspace{.5cm}

\noindent{MSC 2020: 28A35; 28A60; 47A07; 81Q65
  
  % 28A35 Measures and integrals in product spaces
  % 28A60 Measures on Boolean rings, measure algebras
  % 47A07 Forms (bilinear, sesquilinear, multilinear)
  % 81Q65 Alternative quantum mechanics (including hidden variables, etc.)
}

%\tableofcontents

%%%%%%%%%%%%%%%%%%%%%%%%%%%%%%%%
%%%%%%%%%%%%%%%%%%%%%%%%%%%%%%%%
\section*{Introduction}

Consider a \emph{measurable} (or just plain \emph{measure}) \emph{space} $(X,\cA)$, i.e. \cite[Remark 11B(c)]{frml_meas-01_2004} a set $X$ equipped with a \emph{$\sigma$-algebra} \cite[Definition 111A]{frml_meas-01_2004} $\cA\subseteq 2^X$. \cite[\S 2]{MR1303988} interprets the familiar \cite[\S 8.3]{hgh_qm_1989} two-slit setup paradigmatic in quantum mechanics as the non-vanishing of the \emph{interference}
\begin{equation*}
  I_1\mu(S_0,S_1):=\mu(S_0)+\mu(S_1)-\mu(S_0\sqcup S_1)
  \quad
  \left(\text{$S_j$ disjoint}\right). 
\end{equation*}
of a set function $\mu\in \bR^{\cA}$. This in turn suggests a chain of such \emph{interference operators} (\cite[(1)]{MR1303988} or \cite[\S 3]{MR1908598}, with a shift in indices as compared to here)
\begin{equation*}
  I_d\mu\left(S_0,\cdots,S_d\right)
  :=
  \sum_{\ell=0}^{d-1}
  (-1)^{\ell}
  \sum_{(i_0<\cdots<i_{\ell})\in \{0..d\}^{\ell+1}}
  \mu\left(\bigsqcup_k S_{i_k}\right).
\end{equation*}
\cite[\S 5]{MR2728531} terms the vanishing of $I_d$ on a set function $\mu\in \bR^{\cA}$ defined on a $\sigma$-algebra $\cA\subseteq 2^X$ is what \emph{grade-$d$ additivity} and (by analogy to the \emph{countable} (or \emph{$\sigma$-})\emph{additivity} \cite[Deﬁnition 1.3.2]{bog_meas-1-2} of a measure) \emph{grade-$d$ measures} on $(X,\cA)$ are then defined there as 
\begin{itemize}[wide]
\item grade-$d$ additive set functions;
\item \emph{continuous} in the sense of \cite[\S 2, p.682]{MR2728531}:
  \begin{equation}\label{eq:cont.fn}
    \begin{aligned}
      \cA\ni A_n
      \xrightarrow[\quad n\quad]{\quad\text{non-decreasing}\quad}
      A:=\bigcup_n A_n
      \in \cA
      &\xRightarrow{\quad}
        \mu(A)=\lim_n \mu(A_n)
        \quad\left(\text{\emph{upper continuity}}\right)\\
      \cA\ni A_n
      \xrightarrow[\quad n\quad]{\quad\text{non-increasing}\quad}
      A:=\bigcap_n A_n
      \in \cA
      &\xRightarrow{\quad}
        \mu(A)=\lim_n \mu(A_n)
        \quad\left(\text{\emph{lower continuity}}\right).
    \end{aligned}    
  \end{equation}
\end{itemize}
As grade-$d$ measures are assumed $\bR_{\ge 0}$-valued in \cite{MR2728531} but not below, modifiers such as \emph{positive} (i.e. taking non-negative values) and \emph{signed} (real-valued) will occasionally qualify the phrase for clarity. 

\cite[\S\S 4,5]{MR2728531} revolve around tracing connections between grade-$d$ measures and plain measures on the power $(X,\cA)^d=(X^d,\cA^{\otimes d})$ in the usual sense \cite[\S 3.3]{bog_meas-1-2} of measure-space products. Writing
\begin{equation*}
  \cA
  \ni A
  \xmapsto{\quad\wt{\lambda}\quad}
  \lambda(A^d)
\end{equation*}
for the \emph{diagonal} of a function defined on $\cA^d$ (just the plain Cartesian product, consisting of \emph{cylinders} \cite[p.188]{bog_meas-1-2} $A_1\times\cdots \times A_d$, $A_j\in \cA$),

\begin{itemize}[wide]
\item \cite[Theorem 5.1]{MR2728531} argues that $\wt{\lambda}$ is a positive grade-$d$ measure whenever $\lambda$ is a \emph{diagonally positive} (i.e. $\wt{\lambda}\ge 0$) measure on $(X,\cA)^d$;

\item while conversely, \cite[Theorem 4.2]{MR2728531} asserts that a grade-$2$ positive measure $\mu$ is uniquely of the form $\wt{\lambda}$ for a diagonally-positive \emph{symmetric} ($\lambda(A\times B)=\lambda(B\times A)$) measure on $(X,\cA)^2$. 
\end{itemize}

It is the latter statement that provided the motivation for the present note, for it appears to admit counterexamples (hence the paper's and \Cref{se:qmeas}'s titles):

\begin{theorem}\label{th:super.not.diag}
  For every $d\in \bZ_{\ge 2}$ there is a positive grade-$d$ measure $\mu$ on a measure space $(X,\cA)$ not realizable as $\wt{\lambda}$ for a signed diagonally-positive measure $\lambda$ on $(X,\cA)^d$.
\end{theorem}

% % OLD: SHUFFLED TO INTRODUCTION
% % 
% % \cite[Theorem 4.2]{MR2728531}, to the effect that a positive grade-$2$ measure $\mu$ on an arbitrary measure space $(X,\cA)$ is of the form
% % \begin{equation*}
% %   \mu(A)
% %   =
% %   \wt{\lambda}(A)
% %   :=
% %   \lambda(A\times A)
% % \end{equation*}
% % for a \emph{diagonally-positive} \cite[p.692]{MR2728531} signed measure $\lambda$ on $(X,\cA)$, appears to me to be invalid. \Cref{th:super.not.diag} argues that this is so, constructing $\mu$ not of this form (hence the present section's and the draft's titles). 
% % 

The gap in the proof of \cite[Theorem 4.2]{MR2728531} appears to be the jump (relegated in the text to an otherwise unspecified ``standard argument'') between
\begin{itemize}[wide]
\item the \emph{separate} $\sigma$-additivity of the only candidate $\lambda$ with $\wt{\lambda}=\mu$, in the sense that
  \begin{equation}\label{eq:sep.s.add}
    \forall\left(A\in \cA\right)
    \ :\
    \lambda(A\times -),\ \lambda(-\times A)
    \text{ are $\sigma$-additive},
  \end{equation}
  
\item and its (plain, global) $\sigma$-additivity on the \emph{set algebra} \cite[Definition 1.2.1]{bog_meas-1-2}
  \begin{equation*}
    \cA\otimes_0 \cA
    :=
    \left\{\bigcup_{i=1}^n A_i'\times A''_i\ :\ A'_i,\ A''_i\in \cA\right\}
  \end{equation*}
  of finite unions of cylinders.
\end{itemize}
\Cref{eq:sep.s.add} in fact makes the set function $\lambda$ what the literature refers to as a \emph{bimeasure}: \cite[Definition 4.1]{zbMATH03620539}, \cite[p.234]{MR593464}, \cite[\S 1]{MR86116} or \cite[\S 1.2]{MR1849394} (and \cite[Definition 1]{MR904773}, for instance, for the more general, $d$-factor \emph{polymeasures}). The existence of non-$\sigma$-additive polymeasures (in addition to \cite[\S 1 (ii)]{MR59996} see also \cite[Example 2]{MR593464}) at the very least shows that the passage between the two types of countable additivity cannot be automatic.

What the proof of \cite[Theorem 4.2]{MR2728531} appears to (effectively) show is rather

\begin{theorem}\label{th:qmeas.bimeas}
  For a measurable space $(X,\cA)$ the map $\lambda\mapsto \wt{\lambda}$ restricts to a linear isomorphism between
  \begin{itemize}[wide]
  \item the space of finite symmetric signed bimeasures on $(X,\cA)^2$;

  \item and that of signed grade-2 measures on $(X,\cA)$. 
  \end{itemize}
  Under said bijection the diagonal positivity of $\lambda$ becomes equivalent to the positivity of the corresponding grade-2 measure.
\end{theorem}

This in turn relies on an ever so slightly enhanced \cite[Theorem 5.1]{MR2728531}, reading as follows.

\begin{theorem}\label{th:all.poly.grade.d}
  For a measure space $(X,\cA)$ and $d\in \bZ_{\ge 1}$ the diagonal $\wt{\lambda}$ of a signed polymeasure on $(X,\cA)^d$ is a grade-$d$ (signed) measure. 
\end{theorem}

The text below provides short arguments, indicating the few alterations that the existing proofs require. Higher-grade analogues of the above are the subject of future work \cite{chi_sqmeas_pre}.

% % %%%%%%%%%%%%%%%%%%%%%%%%%%%%%%%%
% % \subsection*{Acknowledgments}
% % 
% % I am grateful to D. Rizzolo for bringing \cite{zbMATH03239995} to my attention. 
% %

% % %%%%%%%%%%%%%%%%%%%%%%%%%%%%%%%%
% % %%%%%%%%%%%%%%%%%%%%%%%%%%%%%%%%
% % \section{Preliminaries}\label{se:prel}
% %

%%%%%%%%%%%%%%%%%%%%%%%%%%%%%%%%
%%%%%%%%%%%%%%%%%%%%%%%%%%%%%%%%
\section{Non-diagonal weakly-additive measures}\label{se:qmeas}

\Cref{ex:d.fold.mrn} adapts \cite[\S 1 (ii)]{MR59996} to arbitrary numbers of factors.   \Cref{ex:d.fold.mrn} below, is a straightforward adaptation of \cite[\S 1 (ii)]{MR59996} (which is its 2- rather than $d$-fold analogue), and provides an example of
\begin{itemize}[wide]
\item probability spaces $(Y_i,\cB_i,\nu_i)$, $i\in \{1..d\}$;

\item and a non-negative, additive, non-$\sigma$-additive, total-mass-1 function
  \begin{equation*}
    \cB_1\otimes_0\cdots\otimes_0 \cB_d
    :=\left\{\text{finite unions of $\prod_i \cB_i$-members (\emph{cylinders})}\right\}
    \xrightarrow{\quad \pi\quad}
    [0,1]
  \end{equation*}
  with \emph{marginals} \cite[9.12(vii)]{bog_meas-1-2}
  \begin{equation}\label{eq:nu.marg}
    \pi\left(Y_1 \cdots Y_{i-1}\times \bullet\times Y_{i+1} \cdots Y_d\right)
    =
    \nu_i
    ,\quad
    \forall i\in \{1..d\}.
  \end{equation}
\end{itemize}

\begin{example}\label{ex:d.fold.mrn}
  Set $Y_j:=I:=[0,1]$ for all $j\in \{1..d\}$ and denote by $\ell$ the usual Lebesgue measure on $I$. Fix also a partition
  \begin{equation*}
    I=\bigsqcup_{j=1}^d Z_j
    ,\quad
    \ell^*(Z_j)=1
    ,\quad
    \ell^*:=\text{Lebesgue \emph{outer measure} \cite[\S 1.5]{bog_meas-1-2}},
  \end{equation*}
  and define
  \begin{equation*}
    \cB_j:=\left\{S\cap Z_j\ :\ S\in 2^I\text{ Borel}\right\}
    ,\quad
    j=\in\{1..d\}.
  \end{equation*}
  Finally, for a finite disjoint union $S\in 2^{I^d}$ of Borel cylinders $\pi$ is defined by
  \begin{equation}\label{eq:def.pi}
    \prod_{j=1}^d \cB_j
    \ 
    \ni
    \ 
    T\cap \left(\prod_j Z_j\right)
    \xmapsto{\quad\pi\quad}
    \ell_{diag}(T)
    :=
    \ell\left(\left\{t\in I\ :\ (t,\cdots,t)\in T\right\}\right).
  \end{equation}
  The argument in \cite[\S 1.5]{bog_meas-1-2} plainly applies to show that $\pi$ cannot be countably-additive. Observe also that the marginals \Cref{eq:nu.marg} of $\pi$ are
  \begin{equation*}
    \cB_j
    \ni
    S\cap Z_j
    \xmapsto{\quad\nu_j\quad}
    \ell(S),
  \end{equation*}
  indeed a (countably-additive) probability measure. 
\end{example}

\pf{th:super.not.diag}
\begin{th:super.not.diag}
  In the context of \Cref{ex:d.fold.mrn} (with its notation in place throughout the present discussion), $\mu$ is as follows.
  \begin{itemize}[wide]
  \item Set $X:=\bigsqcup_j Y_j$, with $\cA|_{Y_j}=\cB_j$, $j\in \{1..d\}$ (i.e. form the disjoint union of measure spaces in the most straightforward fashion).

  \item $\pi$ can now be regarded as defined on (a set algebra over) $\prod_j Y_j\subset X^d$, hence also on all of $X^d$ as annihilating subsets of
    \begin{equation*}
      \bigsqcup_{\substack{\text{permutations}\\\sigma\ne\id}}
      \prod_{j=1}^d Y_{\sigma i}
      =
      X^d\setminus \left(Y_1\times \cdots\times  Y_d\right);
    \end{equation*}
    as such, we relabel $\pi$ as $\lambda$.

  \item Finally, take for $\mu$ the diagonal $\wt{\lambda}$:
    \begin{equation*}
      \mu(A)
      :=
      \wt{\lambda}(A)
      :=
      \lambda\left(A\times\cdots  \times A=A^d\right).
    \end{equation*}
  \end{itemize}
  Grade-$d$ additivity is immediate from the additivity of $\lambda$ in each variable, so we will be done once we verify that $\mu$ is continuous in the sense of \Cref{eq:cont.fn}:
  \begin{itemize}[wide]
  \item \cite[Lemma 1.1]{chi_sqmeas_pre} (a higher-grade analogue of \cite[Lemma 4.1]{MR2728531}) recovers the \emph{symmetrization}
    \begin{equation*}
      \frac 1{d!}
      \sum_{\sigma\in \text{symmetric group }S_d}
      \lambda\left(A_{\sigma 1}\times\cdots\times A_{\sigma d}\right)
    \end{equation*}
    from $\wt{\lambda}$ for mutually-disjoint $A_i\in \cA$;

  \item \emph{every}
    \begin{equation*}
      \prod_{j=1}^d A_j\subseteq \prod_{j=1}^d Y_j \subseteq X^d
      ,\quad
      A_j\in \cB_j
    \end{equation*}
    has mutually-disjoint components;

  \item and by construction the restriction of $\lambda$ to $\prod_j Y_j$ is not $\sigma$-additive on $\cB_1\otimes_0\cdots\otimes_0 \cB_d$. 
  \end{itemize}
  It is in verifying the continuity of $\mu$ that the details of how $(Y_j,\cB_j,\nu_j)$ and $\pi$ are constructed become relevant, and we refer the reader to \Cref{ex:d.fold.mrn} for notation.

  I first claim that the set functions $\pi$ of \Cref{ex:d.fold.mrn} are polymeasures (and hence so too is $\lambda$). Indeed, writing $A_{j'}=S_{j'}\cap Z_{j'}$ for $j'\ne j$, simply observe that
  \begin{equation*}
    \pi\left(A_1,\cdots,A_{j-1},\bullet,A_{j+1},\cdots,A_d\right)
    =
    \left(      
      S\cap Z_j
      \xmapsto{\quad}
      \ell\left(S\cap \bigcap_{j'\ne j}S_{j'}\right)
    \right)
    \in
    [0,1]^{\cB_j}:
  \end{equation*}
  plainly countably-additive. \cite[Theorem 1]{MR904773} now provides the continuity of $\lambda$ in the even stronger-than-desired form
  \begin{equation}\label{eq:polymeas.auto.cont}
    A_{jn}
    \xrightarrow[\quad n\quad]{\quad\text{in $\cA$}\quad}
    A_j
    ,\ 
    1\le j\le d
    \quad
    \xRightarrow{\quad}
    \quad
    \lambda\left(\prod_j A_{jn}\right)
    \xrightarrow[\quad n\quad]{\quad}
    \lambda\left(\prod_j A_{i}\right),
  \end{equation}
  where convergence $A_n\xrightarrow[n]{\ } A$ in a $\sigma$-algebra means
  \begin{equation*}
    \bigcup_n \bigcap_{m\ge n}A_m
    =:
    \liminf_n A_n
    =
    A
    =
    \limsup_n A_n
    :=
    \bigcap_n \bigcup_{m\ge n}A_m
  \end{equation*}
  (as in \cite[\S 2.1, Definition 2.1]{gut_prob_2e_2013}, say).
\end{th:super.not.diag}

\pf{th:all.poly.grade.d}
\begin{th:all.poly.grade.d}
  We include a short(er) proof for the grade-$d$ additivity of $\wt{\lambda}$ for polymeasures $\lambda$, to emphasize a useful combinatorial identity between interference operators. Writing
  \begin{equation*}
    \Delta_S \nu
    :=
    \nu-\nu\left(S\sqcup\bullet\right)
    \in \bR^{\left\{T\in \cC\ :\ T\cap S=\emptyset\right\}}
  \end{equation*}
  for real-valued set functions $\nu\in \bR^{\cC}$ defined on some class $\cC\subseteq 2^X$, observe that
  \begin{equation*}
    I_{d-1} \Delta_{S_0}\nu\left(S_1,\cdots,S_d\right)
    =
    I_d \nu\left(S_0,S_1,\cdots,S_d\right)
    -
    \nu(S_0).
  \end{equation*}
  In applying this to $\nu:=\wt{\lambda}$, observe that 
  \begin{equation*}
    I_{d-1} \Delta_{S_0}\wt{\lambda}\left(S_1,\cdots,S_d\right)
    +
    \wt{\lambda}(S_0)
    =
    I_{d-1}\left(
      \Delta_{S_0}\wt{\lambda}-\wt{\lambda}(S_0)
    \right)
    =0
  \end{equation*}
  by induction (on the index $d$), given that the argument $\Delta_{S_0}\wt{\lambda}-\wt{\lambda}(S_0)$ of $I_{d-1}$ (with the second term $\wt{\lambda}(S_0)$ regarded here as a constant set function) is a sum of polymeasures on $(X\setminus S_0,\cA|_{X\setminus S_0})^{k\in \{1..d-1\}}$ obtained by fixing at least one and at most $d-1$ of the arguments of $\lambda$ to $S_0$. 
  
  The only other issue on which the proof of \cite[Theorem 5.1]{MR2728531} needs an update is the fact that even though (generally) not $\sigma$-additive, polymeasures still have continuous diagonals in the sense of \Cref{eq:cont.fn}. \cite[Theorem 1]{MR904773} provides this and more (\Cref{eq:polymeas.auto.cont} above), hence the conclusion.
\end{th:all.poly.grade.d}

\pf{th:qmeas.bimeas}
\begin{th:qmeas.bimeas}
  Once we dispose of the first statement the second, concerning positivity, is clearly a consequence. As to the former, \Cref{th:all.poly.grade.d} proves that $\wt{\bullet}$ does indeed take grade-2 signed values, its injectivity on symmetric finitely-additive set functions on the algebra $\cA\otimes_0 \cA$ of finite cylinder unions follows from the reconstruction formula \cite[Lemma 4.1]{MR2728531}
  \begin{equation}\label{eq:rec.l}
    \cA^2
    \ni
    (A_1,A_2)
    \xmapsto{\quad}
    \frac 12
    \left(
      \mu\left(A\cup B\right)
      +
      \mu\left(A\cap B\right)
      -
      \mu\left(A\setminus B\right)
      -
      \mu\left(B\setminus A\right)
    \right),
  \end{equation}
  and the proof of \cite[Theorem 4.2]{MR2728531} already effectively handles the surjectivity of $\wt{\bullet}$, for it does address \emph{separate} countable additivity.
\end{th:qmeas.bimeas}

We record also the following immediate consequence of \Cref{th:qmeas.bimeas}, which does not appear to be obvious a priori (especially in the signed case). 

\begin{corollary}\label{cor:gr2.bdd}
  A grade-2 signed measure $\mu$ on a measure space $(X,\cA)$ is bounded in the sense that
  \begin{equation*}
    \sup\left\{|\mu(A)|\ :\ A\in \cA\right\}
    <
    \infty.
  \end{equation*}
\end{corollary}
\begin{proof}
  Immediate from \Cref{th:qmeas.bimeas}, given that bimeasures are automatically bounded in this same sense \cite[Theorem 4.5]{zbMATH03620539} (as, in fact, are polymeasures generally: \cite[result (N) and the sentence following it]{MR904773}).
\end{proof}

\begin{remark}\label{re:strng.bdd}
  The countable additivity of $\lambda$ defined by \Cref{eq:rec.l} on $\cA\times \cA$ and extended additively to the algebra $\cA^{\otimes_0 2}$ of finite cylinder unions is reduced in \cite[Th\'eor\`eme 4]{MR450494} over sufficiently well-behaved measure spaces (\emph{universally-measurable} \cite[\S 2]{MR436342} subsets of compact metrizable spaces equipped with the Borel $\sigma$-algebra) to a stronger form of boundedness than that provided by \Cref{cor:gr2.bdd}: one needs rather
  \begin{equation}\label{eq:strng.bdd}
    \sup\sum_i \left|\lambda\left(A'_i\times A''_i\right)\right|<\infty,
  \end{equation}
  the supremum being taken over all finite disjoint families $\{A'_i\times A''_i\}_i$ of cylinders. Now, a bimeasure can ``misbehave'' in at least two qualitatively distinct ways:
  \begin{itemize}[wide]
  \item In \Cref{ex:d.fold.mrn} \Cref{eq:strng.bdd} plainly holds for everything in sight is positive and finite, but the measure spaces involved are not ``sufficiently tame'' in the sense alluded to above.

  \item On the other hand, in \cite[Example 2]{MR593464} that tameness does obtain: the bimeasure is of the form 
    \begin{equation*}
      2^{\left(\bZ_{\ge 0}\right)^2}
      \ni
      (S,T)
      \xmapsto{\quad}
      \sum_{\substack{s\in S\\t\in T}}
      a(s,t)
      \in
      \bR,
    \end{equation*}
    where $\left(a(s,-)\right)_{s\ge 0}$ is a sequence in $\ell^1(\bR)$ \emph{unconditionally} \cite[\S 1]{MR33975} but not \emph{absolutely} convergent (such sequences always exist \cite[Theorem 1]{MR33975} in infinite-dimensional Banach spaces). The tameness referred to above obtains here, so the bimeasures in question can only fail to be measures because \Cref{eq:strng.bdd} does not hold.
  \end{itemize}
  In the language of \cite[Definition 3]{MR904773}, say, the distinction is that between finite bimeasures having finite \emph{semivariation} (which they all do \cite[Theorem 3(4)]{MR904773}) and their having finite \emph{(total) variation}, which (by contrast to measures \cite[Lemma III.1.5 and Corollary III.4.5]{ds_linop-1_1958}) some do not. 
\end{remark}

%%%%%%%%%%%%%%%%%%%%%%%%%%%%%%%%
%%%%%%%%%%%%%%%%%%%%%%%%%%%%%%%%

\addcontentsline{toc}{section}{References}
%\bibliography{bib}{}

\begin{thebibliography}{10}

\bibitem{bog_meas-1-2}
V.~I. Bogachev.
\newblock {\em Measure theory. {V}ol. {I}, {II}}.
\newblock Springer-Verlag, Berlin, 2007.

\bibitem{chi_sqmeas_pre}
Alexandru Chirvasitu.
\newblock Polarized quantum measures, 2025.
\newblock in preparation.

\bibitem{MR1849394}
Nicolae Dinculeanu and Muthu Muthiah.
\newblock Bimeasures in {B}anach spaces.
\newblock {\em Ann. Mat. Pura Appl. (4)}, 178:339--392, 2000.

\bibitem{MR904773}
Ivan Dobrakov.
\newblock On integration in {B}anach spaces. {VIII}. {P}olymeasures.
\newblock {\em Czechoslovak Math. J.}, 37(112)(3):487--506, 1987.

\bibitem{ds_linop-1_1958}
Nelson Dunford and Jacob~T. Schwartz.
\newblock Linear operators. {I}. {General} theory. ({With} the assistence of
  {William} {G}. {Bade} and {Robert} {G}. {Bartle}).
\newblock Pure and {Applied} {Mathematics}. {Vol}. 7. {New} {York} and
  {London}: {Interscience} {Publishers}. xiv, 858 p. (1958)., 1958.

\bibitem{MR33975}
A.~Dvoretzky and C.~A. Rogers.
\newblock Absolute and unconditional convergence in normed linear spaces.
\newblock {\em Proc. Nat. Acad. Sci. U.S.A.}, 36:192--197, 1950.

\bibitem{frml_meas-01_2004}
D.~H. Fremlin.
\newblock {\em Measure theory. {V}ol. 1}.
\newblock Torres Fremlin, Colchester, 2004.
\newblock The irreducible minimum, Corrected third printing of the 2000
  original.

\bibitem{MR436342}
R.~K. Getoor.
\newblock On the construction of kernels.
\newblock In {\em S\'eminaire de {P}robabilit\'es, {IX}}, volume Vol. 465 of
  {\em Lecture Notes in Math.}, pages 443--463. Springer, Berlin-New York,
  1975.
\newblock Seconde Partie, Univ Strasbourg, Strasbourg, ann\'ees universitaires
  1973/1974 et 1974/1975,.

\bibitem{MR2728531}
Stan Gudder.
\newblock Quantum measure theory.
\newblock {\em Math. Slovaca}, 60(5):681--700, 2010.

\bibitem{gut_prob_2e_2013}
Allan Gut.
\newblock {\em Probability: a graduate course}.
\newblock Springer Texts in Statistics. Springer, New York, second edition,
  2013.

\bibitem{MR450494}
Joseph Horowitz.
\newblock Une remarque sur les bimesures.
\newblock In {\em S\'eminaire de {P}robabilit\'es, {XI} ({U}niv. {S}trasbourg,
  {S}trasbourg, 1975/1976)}, volume Vol. 581 of {\em Lecture Notes in Math.},
  pages 59--64. Springer, Berlin-New York, 1977.

\bibitem{hgh_qm_1989}
R.~I.~G. Hughes.
\newblock {\em The structure and interpretation of quantum mechanics}.
\newblock Harvard University Press, Cambridge, MA, 1989.

\bibitem{MR593464}
Igor Kluv\'anek.
\newblock Remarks on bimeasures.
\newblock {\em Proc. Amer. Math. Soc.}, 81(2):233--239, 1981.

\bibitem{MR59996}
E.~Marczewski and C.~Ryll-Nardzewski.
\newblock Remarks on the compactness and non direct products of measures.
\newblock {\em Fund. Math.}, 40:165--170, 1953.

\bibitem{MR86116}
Marston Morse and William Transue.
\newblock {$C$}-bimeasures {$\Lambda$} and their integral extensions.
\newblock {\em Ann. of Math. (2)}, 64:480--504, 1956.

\bibitem{MR1908598}
Roberto~B. Salgado.
\newblock Some identities for the quantum measure and its generalizations.
\newblock {\em Modern Phys. Lett. A}, 17(12):711--728, 2002.

\bibitem{MR1303988}
Rafael~D. Sorkin.
\newblock Quantum mechanics as quantum measure theory.
\newblock {\em Modern Phys. Lett. A}, 9(33):3119--3127, 1994.

\bibitem{zbMATH03620539}
Kari Ylinen.
\newblock On vector bimeasures.
\newblock {\em Ann. Mat. Pura Appl. (4)}, 117:115--138, 1978.

\end{thebibliography}
%\bibliographystyle{plain}

\def\polhk#1{\setbox0=\hbox{#1}{\ooalign{\hidewidth
  \lower1.5ex\hbox{`}\hidewidth\crcr\unhbox0}}}

\Addresses

\end{document}